\documentclass[12pt]{article}
\usepackage{enumerate,amsthm,amsfonts, mathrsfs, amssymb, latexsym, dsfont}
\usepackage{makeidx,setspace, fancyhdr}
\usepackage{amsmath}
\usepackage{avant}
\usepackage{verbatim}
\usepackage{graphicx}
\usepackage{cancel}
\usepackage[left=3cm,right=3cm,top=3cm,bottom=3cm]{geometry}
\usepackage{color}
\usepackage[authoryear]{natbib}

\newtheorem{thm}{Theorem}[section]
\newtheorem{lem}{Lemma}[section]
\newtheorem{rem}{Remark}[section]
\newtheorem{cor}{Corollary}[section]
\newtheorem{df}{Definition}[section]

\newcommand{\beqo}{\begin{eqnarray*}}
\newcommand{\eeqo}{\end{eqnarray*}\noindent}
\newcommand{\beq}{\begin{eqnarray}}
\newcommand{\eeq}{\end{eqnarray}\noindent}

\newcommand{\R}{\mathbb{R}}

\newcommand{\Eop}{\mathbb{E}}
\newcommand{\Pop}{\mathbb{P}}
\newcommand{\Qop}{\mathbb{Q}}
\newcommand{\F}{\mathbb{F}}
\newcommand{\FC}{\mathcal{F}}

\numberwithin{equation}{section}

\begin{document}

\title{\textbf{Optimal investment under partial  information and robust VaR-type constraint }}

\author{ \textbf{Nicole B\"auerle}$^{\S}$,\ \textbf{An Chen}$^*$\\
\
\\
\\
{\small $^{\S}$\emph{Karlsruhe Institute of Technology, Department of Mathematics}}
\\
{\small \emph{ D-76128 Karlsruhe, Germany}}
\\
{\small \emph{nicole.baeuerle@kit.edu}}
\\
{\small \textit{$^*$University of Ulm, Institute of Insurance Science}}
\\
{\small \emph{Helmholtzstrasse 20, D-89069 Ulm, Germany}}
\\
{\small \emph{an.chen@uni-ulm.de}}
}

\date{\today}

\maketitle{}

\begin{abstract}


\noindent This paper extends the utility maximization literature by combining partial information and (robust) regulatory constraints. Partial information is characterized by the fact that the stock price itself is observable by the optimizing financial institution, but the outcome of the market price of the risk $\theta$ is unknown to the institution.
The regulator develops either a congruent or distinct perception of the market price of risk in comparison to the financial institution when imposing the Value-at-Risk (VaR) constraint. We also discuss a robust VaR constraint in which the regulator uses a worst-case measure.
The solution to our optimization problem takes the same form as in the full information case: optimal wealth can be expressed as a decreasing function of state price density. The optimal wealth is equal to the minimum regulatory financing requirement in the intermediate economic states.  
The key distinction lies in the fact that the price density in the final state depends on the overall evolution of the estimated market price of risk, denoted as $\hat{\theta}(s)$ or that the upper boundary of the intermediate region exhibits stochastic behavior.
\\

\noindent JEL: C6, G1, D9.\\

\noindent Keywords: Uncertainty about drift, Value-at-Risk-based regulation, Risk Management 

\end{abstract}

\section{Introduction}

Dating back to \cite{merton1969, merton1971},  utility maximization and optimal asset allocation problems have become one of the most studied optimization problems in economics, finance and insurance. Merton has made two crucial assumptions to achieve analytical solutions: First, the risk preferences of the optimizing agent are characterized by hyperbolic absolute risk aversion (HARA), such as power and exponential utility.\footnote{Analytical solutions can still be achieved if the HARA utility is extended to so-called symmetric asymptotic HARA utility functions, which incorporate exponential and power utility as limiting cases (see e.g. \cite{chen2011}.}  The second assumption concerns the underlying risky investment which generates normally or log-normally distributed risky payoffs. There is abundant academic research extending Merton's pioneering works. Here, we mention three streams of extensions which are relevant for our study: a) Risk-based capital requirements have been incorporated to Merton's pioneering works in extensive follow-up scholarly research (\cite{basak2001value}, \cite{sass2010},  \cite{chen2018b}). b) 
When modeling the dynamics of the underlying assets, researchers have grappled with the challenge of the unobservability of the drift term of the assets (or the market price of risk), despite the ability to observe asset prices (see e.g. \cite{gennotte1986optimal}, \cite{karatzas2001bayesian}, \cite{honda2003optimal}, \cite{sass2004optimizing}, \cite{rieder2005portfolio}, \cite{brendle2006portfolio}, \cite{bjork2010optimal}, \cite{lindensjo2016optimal}). Note that \cite{sass2010} already combine partial information and capital requirements, but use an expected loss criterion, an HMM market and no ambiguity.  c) In case of model ambiguity, robust approaches have been pursued which take into account a number of different models.
Following the approach with multiple priors of \cite{gilboa2004maxmin}, \cite{schmeidler1989subjective}
there have been numerous suggestions to consider robust utility maximization problems
(\cite{quenez2004optimal}, \cite{schied2005optimal}, \cite{ismail2019robust}, \cite{sass2021robust}, \cite{bauerle2021distributionally}).
As the risk constraint, the incomplete information and robustness in the financial market are of vital importance when considering the asset allocation problem of a financial institution, the present paper combines these three aspects and studies their combined impact on the asset allocation. \\

\noindent In the present paper, we analyze the optimal asset allocation problem of a financial institution with partial information subject to various Value-at-Risk (VaR)-type regulations. In practice, VaR has been frequently used as a regulatory and internal risk management tool in the financial industry to compute regulatory or economic capital. VaR can be considered as the cornerstone of the Basel III and Solvency II capital requirements and is widely used. Further, we incorporate parameter uncertainty in the optimization problem of the financial institution. We assume that the institution can observe the asset prices on the market, but cannot observe the drift term of the assets (or the market price of risk). This assumption is realistic since the volatility of the stock prices can be estimated very well, whereas the drift is notoriously difficult to estimate (see e.g. \cite{gennotte1986optimal}).  Due to the observed partial information, both the financial institution and regulator shall build beliefs about the market price of risk. \\

\noindent In \emph{one} setting, we consider the case in which the regulator 
has the same information about the underlying asset as the optimizing financial institution. 
In other words, we  assume that the regulator is  perfectly certain about the
objective probability law of the state process and this belief is identical to the one used by the optimizing financial institution. 
However, this assumption does not always hold, in particular when the firm itself only owns partial information about the asset dynamic and needs to build own beliefs about the market price of risk.  Generally, the nullification of this assumption can be justified by the fact that the financial institution has some firm-specific information and the regulator is only able to use information external to the
firm, for instance, some industry-wise measures.\footnote{\cite{chen2009knightian} discuss insurance regulation problems where the regulator may not have perfect confidence
on the perceived probability measure for insurance companies’ future asset
value.} Hence, in our \emph{second} model setup, we look at the case that regulators, as outsiders, may have a different perceived probability to determine the Value-at-Risk. Alternatively, regulators may want to be on the safe side and use a robust approach.\\

\noindent There are numerous empirical studies documenting the impact of general economic policy uncertainty (EPU) on asset allocation and asset prices (\cite{julio2012political}; \cite{kelly2016price}).  Our paper enriches this stream of literature by  theoretically studying the impact of regulatory uncertainty on the asset allocation. It can be considered as direct extensions to \cite{basak2001value} (for the setting that the regulator and financial institution know the market price of risk) and \cite{gabih2009utility} (for the setting that the regulator and financial institution 
build various beliefs about the market price of risk). \\

\noindent The optimal terminal wealth under partial information and VaR-type regulatory constraint turns out to take the same form as in the full information case (c.f. \cite{basak2001value} and \cite{gabih2009utility}):  
It displays a 3-region solution. In good and bad economy (with small or large state price density realizations), the solution is a Merton-type form, while in the intermediate economic states, the solution corresponds to the minimum regulatory financing requirement. 
Furthermore, the solution suggests lower optimal terminal wealth levels than the unconstrained solution with no VaR in rather favorable and unfavorable economic state, while the optimal terminal wealth outperforms the unconstrained solution in the intermediate economic state.  
Despite the same form, the optimal solution under partial information is more complex than the full information case, as the terminal state price $\xi_T$ depends on the entire evolution of the estimated market price of risk $\hat{\theta}(s)$ between $[0,T]$, i.e., on the entire path of the stock price. In contrast, in the full information case (solution in \cite{basak2001value}), the optimal terminal wealth is exclusively a function of the terminal state price density (or the terminal stock price).  \\

\noindent  In a more realistic environment, the regulator will form a different prior for the market price of the risk than the financial institution, resulting in a VaR constraint under a different probability measure. The regulator may also use a worst-case measure to impose more stringent capital requirements. 
This difference in the regulatory constraint does not change the overall structure of the optimal terminal wealth. In adverse economic situations, default cannot be avoided with certainty, i.e., the financial institution ends up with terminal assets that are smaller than the minimum regulatory funding requirement. In addition, the optimal terminal wealth has an intermediate range in which the regulatory minimum funding requirement is met. In contrast to the scenario where the regulator shares identical priors regarding the market price of risk, the upper boundary of the intermediate region now exhibits stochastic behavior.\\



\noindent The remainder of the paper is organized as follows. Section \ref{sec:market} describes the underlying financial market. 
Sections \ref{sec:opt_VaR&PI} and \ref{sec:PI&VaR_robust} 
study and solve the optimal terminal wealth problem for partial information under various VaR-type regulations: Section \ref{sec:opt_VaR&PI} deals with a situation where all have the same belief and a VaR constraint as in \cite{basak2001value}. We determine the optimal terminal wealth and  the optimal investment strategy. Section  \ref{sec:PI&VaR_robust} studies a robust VaR constraint.
A series of numerical analyses is conducted in Section \ref{sec_num_ana}. Finally, Section \ref{sec:con} provides some concluding remarks and perspectives for further research, and several proofs are collected in Section \ref{sec:app}. 

\section{Underlying financial market }\label{sec:market}
We assume that we are in a financial market with two traded assets: one risky and one risk-free asset. Suppose that $(\Omega, \FC, \F=\{\FC_t, 0\le t  \le T\}, \Pop)$ is a filtered probability space and $T$ is the final time horizon. The risk free asset $B=(B(t))$ evolves according to
\[
d B(t) \ = \ r B(t) dt, \quad  t\in[0,T]
\]
for a constant risk free rate $r\ge 0$. The asset price dynamics for the risky
asset $S=(S(t))$ is given by
\begin{align}\label{stockprice} d S(t) =& S(t) \left( (r + \theta \sigma )dt + \sigma d W(t)\right),  \quad   t\in[0,T],
\end{align}
where $W=(W(t))$ denotes a standard Brownian motion under the probability
measure $\mathbb{P}$ and $\sigma>0$ the volatility of the stock.
The stock price itself is observable to the investor, but the outcome of the market price of risk $\theta$ is unknown to her. It is a random variable with known initial distribution $\Pop(\theta=\vartheta_k) =: p_k>0,\;
k=1,\ldots ,m$ where $\vartheta_1,\ldots ,\vartheta_m$ are the possible values of $\theta$. 
This distribution can be viewed as a prior belief that the investor has about the distribution of the market price of the risk at time $t=0$.
The distribution of  $\theta$ will be updated by the investor based on her information. We assume that $\theta$ and $W$ are independent. Note that observing the stock price is equivalent to observing the process $Y$ with $Y(t):=W(t) + \theta t $, as we can rewrite the stock price as follows:
\begin{align}\label{stockprice1} d S(t) =&  S(t) \left(r dt + \sigma  d Y(t)\right).
\end{align}
In the considered Bayesian model, the investor draws her inferences about $\theta$ and updates her belief about $\theta$ via $\hat{\theta}(t) = \Eop[\theta  | \FC^S_t] $ where  $\F^S=\{\FC^S_t, 0\le t   \le T\}$ is the augmented filtration generated by the stock price processes $S$. This conditional expectation can be computed as follows. First, let us consider $\Pop(\theta=\vartheta_k  | \FC^S_t)$. Standard filtering theory implies that $ \Pop(\theta=\vartheta_k  | \FC^S_t)=\Pop(\theta=\vartheta_k  | Y(t))$.\footnote{The process $p_k(t) := \Pop(\theta=\vartheta_k  | \FC^S_t)$ is called Wonham-filter and an SDE can be derived for it, see e.g. \cite{elliott2008hidden}, \cite{karatzas2001bayesian}.} Using Bayes' rule and the fact that the density of $W(t)$ is given by $\varphi_t(x) = (2\pi t)^{-1/2} e^{- x^2/(2t)}$, we obtain that
\begin{align}
\nonumber \Pop(\theta=\vartheta_k  | Y(t)= y)=& \frac{\Pop(\theta=\vartheta_k, Y(t)= y)}{\sum_{i=1}^m\Pop(\theta=\vartheta_i ,Y(t)= y)} = \frac{\Pop(\theta=\vartheta_k, W(t)= y-\vartheta_k t)}{\sum_{i=1}^m\Pop(\theta=\vartheta_i ,W(t)= y-\vartheta_it)} \\
=&  \frac{p_k \varphi_t(y-\vartheta_kt)}{\sum_{i=1}^m p_i \varphi_t(y-\vartheta_i t)},
\end{align}
where $\Pop(X=x)$ should be understood as the density of $X$. Inserting now the density and eliminating all common factors finally yields the expression
$$\Pop(\theta=\vartheta_k | Y(t))= \frac{p_k L_t(\vartheta_k,Y(t))}{F(t,Y(t))},$$
where
$$ F(t,y) := \sum_{k=1}^m L_t(\vartheta_k,y) p_k, \quad L_t(\vartheta_k,y) := \exp(\vartheta_k y -\frac12  \vartheta_k^2 t).$$
Hence, the conditional expectation is given by  (here and later we denote by $F_y$ the derivative of $F$ w.r.t. $y$)
\begin{align*}\hat{\theta}(t) := \Eop[\theta  | \FC^S_t] = \sum\limits_{k=1}^m \vartheta_k \Pop(\theta=\vartheta_k | Y(t)) = \frac{\sum\limits_{k=1}^m\vartheta_k p_k L_t(\vartheta_k,Y(t)) }{F(t,Y(t))}=\frac{F_y(t,Y(t))}{F(t,Y(t))}.
\end{align*}
Based on her own Bayesian updating $\hat{\theta}$, we can rewrite the stock price evolution as
\begin{align}\label{stockpriceOBS} d S(t) =& S(t) \left((r+ \sigma \hat{\theta}(t))dt + \sigma d \hat{W}(t)\right), \end{align}
where
\begin{align}
\label{relationBM}   d \hat{W}(t) =&  d W(t) + (\theta-\hat{\theta}(t)) dt, \quad \text{or} \quad d \hat{W}(t) = d Y(t) -\hat{\theta}(t) dt
\end{align}
and thus we ensure that the stock price process expressed in \eqref{stockpriceOBS} agrees with the one in \eqref{stockprice}. Note that L\'evy's characterization of Brownian motion and filtering theory  implies that $\hat{W}=(\hat{W}(t))$ is an $(\F^S,\Pop)$-Brownian motion (see \cite{karatzas2001bayesian}, sec.3 and the references given there). In this sense, the stock price is expressed with all processes being $\F^S$-adapted. It is the key result used later in our optimization problem. Also note that the financial market is here complete.

\section{Optimal solution under partial information and VaR constraints} \label{sec:opt_VaR&PI}
Let $\pi(t)$ denote the amount that the institution invests in the risky asset and $(X(t)-\pi(t))$ the amount invested in the risk-free bank account. We assume that the investment strategy can only be chosen from the following admissible set
\begin{align}
\nonumber \mathcal{A}(x_0) :=&  \Big\{\pi=(\pi(t))\; |\; X_0=x_0,  \; \pi(t) \; \text{is}\; \F^S -\text{progressively measurable},\\
 &   X(t) \geq 0 \; \text{for all} \; t\geq 0, \;\int_0^{  T} \pi^2(t) d t < \infty\; \Pop-a.s.\Big\},
\end{align}
where $x_0$ is given.
Based on these assumptions, the wealth process $X=(X(t))$  of the institution satisfies the SDE
\begin{align} \label{wealthprocess1}
\nonumber  d X(t) = &  \pi(t)\frac{dS(t)}{S(t)} +(X(t)-\pi(t))\frac{dB(t)}{B(t)}  \\
\nonumber  =&\pi(t)\big((r+ \sigma \hat{\theta}(t)) dt+\sigma d \hat{W}(t)\big)+(X(t)-\pi(t)) rdt \\
\nonumber 
=& \big( r X(t) +\pi(t)\sigma \hat \theta(t)\big) dt + \pi(t) \sigma d \hat W(t)\\
 =& r X(t) d t + \pi(t)\sigma d Y(t)  \\ \nonumber 
 X(0)=& x_0.
 \end{align}
The second equality is obtained by inserting the asset price dynamics and the third by rearranging the terms $dt$ and $d\hat W(t)$. In the last step, we use the definition of $Y.$
 In order to stress that strategy $\pi$ is used here, we also write $(X^\pi(t)).$\\

\noindent We assume that our investor solves the following optimization problem
\begin{align} \label{objective}
& \sup_{ \pi \in \mathcal{A}(x_0)  }\;
{\mathbb{E} }\left[U(X^\pi(T))  \right], \;  \;\text{s.t.}\;   \eqref{wealthprocess1} \; \;\text{and}\; \;\Pop (X^\pi(T) \ge L ) \geq 1-\beta,  \;
\end{align}
where  $L>0$ and $\beta\in [0,1]$. The constraint is equivalent to  the VaR constraint  $VaR_\beta \le x_0-L$ where $VaR_\beta$ is defined as
$$ \Pop(x_0-X^\pi(T) \le VaR_\beta )=1-\beta,$$
which is the loss which is exceeded with some probability $\beta.$  $L$ typically represents a regulatory (minimum) funding requirement. It  is usually  the initial investment of the debt-holders of an institutional investor $L_0<x_0$, accrued with a minimum interest $L=L_0 e^{g T}$ with $g<r$. See e.g.\ \cite{broeders2010pension} for further discussions about regulatory thresholds. The investor's utility function $U$ is assumed to be defined on the positive real line,
 is twice differentiable and satisfies the usual Inada conditions: $\lim\limits_{x\downarrow 0 } U^{\prime}(x)=\infty$ and $\lim\limits_{x\uparrow \infty } U^{\prime}(x)=0$. 
Note that the case $\beta=1$ corresponds to the classical terminal utility problem without constraints. The case $\beta=0$ corresponds to portfolio insurance where a minimal capital $L$ is guaranteed.\\

\noindent We aim to find the optimal self-financing investment strategies  to maximize the expected utility from the terminal wealth, given a VaR-type risk constraint. The key to solve the optimization problem under partial information is to express all processes such that they are $\F^S$-adapted and that the market is complete. Thus, the problem can be turned into one with complete information.
Note that, due to the relation
$Y(t) := \hat{W}(t)+\int_0^t\hat{\theta}(s)ds$, we can define a new probability measure $\Qop$, under which $(Y(t))$ becomes a standard $(\F^S,\Qop)$ Brownian motion:
\begin{align}\label{eq:QP}
\frac{d \Qop}{d \Pop} = \exp\left(-\int_0^T \hat{\theta}(s) d \hat{W}(s) -\frac12 \int_0^T \hat{\theta}^2 (s)ds \right) := \nu_T.
\end{align}
Hereby, we have used the fact that $\hat{W}=(\hat{W}(t))$ is a standard Brownian motion under $(\F^S,\Pop)$.\footnote{We do not use  $ W(t)+\theta t$ to define the change of measure, as $\theta$ is not observable. } A proof of the next lemma can be found in the appendix of \cite{bauerle2019optimal}.
\begin{lem}\label{lem:nu}
We can express $\nu_t:= \Eop[\nu_T| \mathcal{F}_t^S]$ as a function of the observational magnitude $Y(t)$:
\begin{align}
\nu_t = F(t,Y(t))^{-1},\quad   t\in[0,T].
\end{align}
\end{lem}

\noindent The discounted process  $\xi_t := e^{-r t} \nu_t$  gives the state price density process in the complete market and satisfies
\begin{align*}
d \xi_t =  \xi_t (-r d t -\hat{\theta}(t) d \hat{W}(t)), \; \xi_0=1.
\end{align*}

\noindent The dynamic optimal asset allocation problem \eqref{objective} is now ready to be solved through the so-called static martingale approach (c.f. \cite{cox1989optimal}). By the  completeness of the market, Problem \eqref{objective} can be rewritten as
\begin{align} \label{eq:optproblem1}
\max_{X^{\pi}(T)} \, {\mathbb{E}} [U(X^{\pi}(T))] \quad \text{s.t.}\quad  \Pop (X^{\pi}(T) \ge  L )\ge 1- \beta \quad \text{and}\quad {\mathbb{E}} [\xi_T X^{\pi}(T)]\le x_0,
\end{align}
where we maximize over $\FC^S_T$ measurable variables.
We denote by $I$ the inverse of the first derivative of the utility function $U^{\prime}$.
We furthermore assume the following integrability condition:

\vspace{2mm}
{\em
\noindent{\it Assumption (A)}: For any $\lambda_1\in(0,\infty)$, we have
\begin{align*}
{\mathbb{E} }[U(I(\lambda_1\xi_T))]<\infty \quad \mbox{and} \quad  {\mathbb{E}} [\xi_T I(\lambda_1\xi_T)]<\infty.
\end{align*}
}

{{\raggedleft{Solving}}} problem \eqref{eq:optproblem1} we have to look at the following Lagrangian
\begin{align}\label{eq:Lag}
U(X^{\pi}(T)) + \lambda_1(x_0- \xi_T X^{\pi}(T) ) + \lambda_{2}(\beta -{\bf 1}_{X^{\pi}(T)<L} )
\end{align}
where $\lambda_1$, $\lambda_2$ are two multipliers for the budget and VaR-type risk constraint, and ${\bf 1}_{A}$ stands for the indicator function which gives 1 if $A$ occurs and $0$ else.
In order to make the above problem admissible, we shall assume that $x_0\geq {\mathbb{E}}[\xi_T L {\bf 1}_{\xi_T\le \bar{\xi}}]$, where $\bar{\xi}$ will be introduced below.
In a full information setting, \cite{basak2001value} solve the above problem. Thanks to the transformation to a complete market setting with same filtration, the optimal terminal wealth here takes the same form:
\begin{align} \label{optTW}
X^*_{T}(\lambda_1,\xi_T) =  I(\lambda_1 \xi_T){\bf 1}_{\xi_T<\underline{\xi}} +L {\bf 1}_{\underline{\xi} \leqslant \xi_{T}< \bar{\xi}} + I(\lambda_1  \xi_T) {\bf 1}_{\bar{\xi}\le \xi_T},
\end{align}
where $I$ is the inverse marginal utility. In order to emphasize the dependence on the Lagrangian multiplier and the state price, we have used the notation $X_T^*(\lambda_1,\xi_T)$. The upper bound $\bar{\xi}$ is defined such that $\Pop(\xi_{T}>\bar{\xi}) = \beta$.
The lower bound $\underline{\xi}$ is defined by $I(\lambda_1 \underline{\xi}) = L$ and $\lambda_1$ is chosen such that the budget constraint is satisfied.
The optimal terminal wealth under VaR-type constraint corresponds to the Merton solution from above (with a different Lagrangian multiplier $\lambda_1$) until the terminal wealth hits $L$ which corresponds to good market scenarios $\xi_T (\omega)<\underline{\xi}$ (``good'' as $\xi_T (\omega)$ is small). For bad scenarios where $\xi_T (\omega)> \bar{\xi}$, the optimal wealth again coincides with the Merton solution. For intermediate scenarios where $\underline{\xi} \leqslant \xi_{T} (\omega) < \bar{\xi}$, the terminal wealth is kept (hedged) at the level $L$ in order to satisfy the VaR constraint. As the budget constraint still needs to hold, the investor pays for this by having a lower wealth in the scenarios where $\xi_T (\omega) < \underline{\xi}$ and $\xi_T (\omega)\geq \bar{\xi}$. The last set of scenarios corresponds to the bad economic states where the value of the Arrow-Debreu securities is high and the stock price is low. Hence, an investor following a VaR optimal strategy will encounter higher losses in deep recessions than an investor in the Merton problem who is not regulated at all. In our setting we obtain a similar solution. We summarize our findings in the next theorem.

\begin{thm}\label{theo:VaR}
The optimal terminal wealth of problem \eqref{eq:optproblem1} is given by
\begin{align} \label{optTW2}
X^*_{T}(\lambda_1,\xi_T) =  I(\lambda_1 \xi_T){\bf 1}_{\xi_T<\underline{\xi}} +L {\bf 1}_{\underline{\xi} \leqslant \xi_{T}< \bar{\xi}} + I(\lambda_1  \xi_T) {\bf 1}_{\bar{\xi}\le \xi_T},
\end{align}
where $\underline{\xi}  = U'(L)/\lambda_1$, $\bar{\xi}$ is such that $\Pop(\xi_T > \bar \xi  )=\beta $ and $\lambda_1$ is the unique solution of ${\mathbb{E}} [\xi_T X_T^*]= x_0.$
\end{thm}

\begin{rem}
Note that $\Pop(\xi_T > \bar \xi  )=\beta $ can more explicitly be written as
\begin{eqnarray*}
\Pop(\xi_T > \bar \xi  )&=& \Pop((e^{rt} F(T,W(T)+\theta T))^{-1} > \bar \xi  )\\
&=& \sum_{k=1}^m p_k \Pop((e^{rt} F(T,W(T)+\vartheta_k T))^{-1} > \bar \xi  )=\beta.
\end{eqnarray*}
Thus, $\bar \xi$ is the first variable which can be determined from this equation prior to computing $\lambda_1.$ 
\end{rem}

\noindent Although the solution to the optimal terminal wealth under partial information in \eqref{optTW2} takes the same form as in the full information case, the optimal solution is more complex. Note that $\xi_T$ depends on the entire evolution of  $\hat{\theta}(s)$ between $0$ and $T$, in other words, on the entire path of $S(t)$ or $Y(t)$. In the full information case (Bask-Shapiro solution), the optimal terminal wealth is exclusively a function of the terminal state price density (or the terminal stock price). Note that Theorem \ref{theo:VaR} is a \emph{special case} of Theorem \ref{theo:VaR2}. The proof of Theorem \ref{theo:VaR2} is postponed to the \emph{appendix}.\\

\noindent We are also able to compute an optimal investment strategy for some specific utility functions explicitly. For what follows we consider the utility  $u(x)=\frac{x^{1-\gamma}}{1-\gamma}$, with $\gamma >0$ and $\gamma \neq 1$ and its inverse marginal utility $I(x)=x^{-1/\gamma}.$ Let further 
\begin{equation}\label{eq:h1line}
    h(t,y):= e^{-rT} \big(h^1(t,y)+ h^2(t,y)+ h^3(t,y)\big),
\end{equation}  
with
\begin{eqnarray*}
h^1(t,y) &:=& \lambda_1^{-1/\gamma} e^{rT/\gamma} \int_{(F^{-1}(T,1/e^{rT}\underline{\xi}) -y)/\sqrt{T-t}}^\infty F(T,y+z\sqrt{T-t})^{1/\gamma}\varphi(z) dz,\\ \nonumber
h^2(t,y) &:=& L \int_{(F^{-1}(T,1/e^{rT}\bar{\xi}) -y)/\sqrt{T-t}}^{(F^{-1}(T,1/e^{rT}\underline{\xi}) -y)/\sqrt{T-t}}
\varphi(z) dz,\\ \nonumber
h^3(t,y) &:=& \lambda_1^{-1/\gamma} e^{rT/\gamma} \int_{-\infty}^{(F^{-1}(T,1/e^{rT}\bar{\xi}) -y)/\sqrt{T-t}} F(T,y+z\sqrt{T-t})^{1/\gamma}\varphi(z) dz,
\end{eqnarray*}
where $\varphi$ is the density of the standard normal distribution. Having defined this we can state an optimal investment strategy as follows:

\begin{thm}\label{theo:VaRoptstrategy}
Assume that the potential values of $\theta$ are all non-negative.
An optimal investment strategy of problem \eqref{eq:optproblem1}  which attains $X^*_{T}(\lambda_1,\xi_T)$ from Theorem \ref{theo:VaR}  is for general utility function given by 
\begin{align} \label{optPS2g}
\pi^*(t) = \frac{e^{rt}}{\sigma} \psi_t, \quad  t\in[0,T]
\end{align}
 where $\psi$ exists and is such that
\begin{equation}
    \frac{X_t^*}{B(t)}= x_0 + \int_0^t \psi_t dY_t.
\end{equation}
For example, in case of power utility we obtain:
\begin{align} \label{optPS2}
\pi^*(t) = \frac{e^{rt}}{\sigma} h_y(t,Y(t)), \quad  t\in[0,T]
\end{align}
with $h$ as in \eqref{eq:h1line}.
\end{thm}
\begin{proof}
Proof is given in Appendix \ref{sec:proof1}.
\end{proof}

\section{Optimal solution under partial information and robust VaR constraints} \label{sec:PI&VaR_robust} 
In this section, we consider first a more general problem than in the previous section. 
Namely we assume that the VaR constraint may be satisfied under a different prior belief, i.e. a different probability measure. 
Thinking of the motivating example in the introduction about a financial institution's asset allocation problem,  regulators might not be sure of the future asset evolution of each financial institution, and hence, may not have perfect confidence in the perceived probability measure for the optimizing financial institution's future asset
value. So suppose there is a second probability measure $\tilde \Pop$ equivalent to $\Pop$ where $\tilde \Pop(\theta=\vartheta_k) =: q_k>0,\;
k=1,\ldots ,m$. We denote the density by  $\tilde\eta_T:=\frac{d\tilde \Pop}{d\Pop}$ and $\eta_T :=  {\mathbb{E}}[\tilde\eta_T | \FC_T^S]$ is the density conditioned on the observable filtration. 
The Brownian motion $W$ is again a Brownian motion under $\tilde \Pop.$
Instead of problem \eqref{eq:optproblem1} we now consider
\begin{align} \label{eq:optproblem2}
\max_{X(T)} \, {\mathbb{E}}[U(X(T))] \quad \text{s.t.}\quad  \tilde\Pop(X(T) \ge L )\ge 1- \beta \quad \text{and}\quad {\mathbb{E}} [\xi_T X(T)]\le x_0.
\end{align}
where $\Pop$ is replaced by $\tilde \Pop$ in the VaR constraint and where we maximize over $\FC^S_T$ measurable variables. If $\tilde \Pop=\Pop$, we are back in the setting of the previous section. Extending the proof of \cite{basak2001value,sass2010} we can see that the optimal terminal wealth is now given by (the proof can be found in Appendix \ref{sec:proof2}).

\begin{thm}\label{theo:VaR2}
The optimal terminal wealth of problem \eqref{eq:optproblem2} is given by
\begin{align} \label{optTW3}
X^*_{T}(\lambda_1,\lambda_2,\xi_T,\eta_T) =  I(\lambda_1 \xi_T){\bf 1}_{\xi_T<\underline{\xi}} +L {\bf 1}_{\underline{\xi} \leqslant \xi_{T}< \bar{\xi}} + I(\lambda_1  \xi_T) {\bf 1}_{\bar{\xi}\le \xi_T},
\end{align}
where $\underline{\xi}  = U'(L)/\lambda_1$ and  $\bar{\xi} := \inf\{ z>U'(L) : f(\lambda_1z)=\lambda_2\eta_T+U(L)\} $ with $f(z):= U(I(z))-zI(z)+Lz$.  Further $\lambda_1,\lambda_2$ are the solution of ${\mathbb{E}} [\xi_T X_T^*]= x_0$ and $\tilde\Pop(X_{T}^* \ge L )= 1- \beta$ where we assume that the parameters are such that the solution exists. The corresponding optimal investment strategy can be  characterized as in Theorem \ref{theo:VaRoptstrategy}. 
\end{thm}



\noindent Please note that, in contrast to Theorem \ref{theo:VaR}, the upper boundary of the intermediate region exhibits stochastic behavior.
We can now explore various scenarios where the constraints are approached from different perspectives. Specifically, given the nature of the optimal solution, we can proceed to examine the problem within the framework of \emph{robust} VaR constraints. In the context of robust VaR constraints, we understand the condition as follows: $\min_{i\in I}\tilde\Pop_i(X(T) \ge L )\ge 1- \beta,$ where $ I:=\{1,\ldots, m\}.$ Consequently, instead of relying on a single probability measure, the regulator considers a set of possible measures and mandates that the constraint holds true for all measures within this set. We address this scenario in Corollary 4.1. However, before presenting it, we require the following definition.

\begin{df} Two random variables $\theta, \theta'$ with values in $\R$ satisfy $\theta \le_{FSD} \theta'$ if and only if 
$$\Eop [f(\theta)]\le \Eop[ f(\theta')]$$
for all non-decreasing $f:\R\to\R$ for which the expectations exist.
\end{df}

\noindent For what follows we assume that the potential values of $\theta$ are all non-negative and  w.l.o.g.\ ordered as $0\le \vartheta_1\le \ldots \le \vartheta_m.$ Then we obtain (for a proof see  Appendix \ref{sec:proof3}):

\begin{cor}\label{cor:1}
Let $\theta_1 \ge_{FSD} \ldots \ge_{FSD} \theta_m \ge_{FSD} \theta $ with corresponding probability measures $\tilde \Pop_i, i\in I:=\{1,\ldots,m\}$ and $\Pop.$
Suppose we replace the VaR constraint in  \eqref{eq:optproblem2} by 
    \begin{align}
     \min_{i\in I}\tilde\Pop_i(X(T) \ge  L )\ge 1- \beta, \mbox{ with } i\in I.
\end{align}
Then the optimal terminal wealth is again given by Theorem \ref{theo:VaR}. The same is true when we replace the VaR constraint by 
 \begin{align} \label{eq:smooth}
     \sum_{i\in I}\alpha_i\tilde\Pop_i(X(T) \ge L )\ge 1- \beta.
\end{align}
where $0\le \alpha_i$ and $\sum_{i\in I} \alpha_i=1.$

\end{cor}

\section{Numerical analyses} \label{sec_num_ana}

In this section, we carry out numerical analyses aimed at assessing the influence of VaR constraints on the optimal terminal wealth, building upon the insights provided by Theorem \ref{theo:VaR}.
In the preceding sections, we have establish the optimal terminal wealth under broad utility function formulations. However, for our numerical examination, we adopt the assumption that the optimizing institution employs a power utility function to express her preferences, i.e., $u(x)=\frac{x^{1-\gamma}}{1-\gamma}$, with $\gamma >0$ and $\gamma \neq 1$. The inverse marginal utility is given by $I(x)=x^{-1/\gamma}$. We fix the parameters as follows:  
\begin{align} \label{para_choice}
\nonumber r =& 0.03, \,T = 10, \,x_0 = 100, \, 
\gamma =2\; \text{or} \;3, \,\beta = 0.05, \, L = 120, \\
\vartheta_1 =& 0.15,\, \vartheta_2 = 0.25, \, \vartheta_3 = 0.35, \,  p_1 = 1/3, \, p_2 = 1/3, \,p_3 =1/3.
\end{align}
The chosen $\vartheta_1$, $\vartheta_2$ and $\vartheta_3$ are in line with a reasonable market price of risk (Sharpe ratio).   
The Sharpe ratio describes the performance of a risky project in relation to a risk-free investment. Typically, a Sharpe ratio above 0.5 in the long run indicates great investment performance and is difficult to achieve, while a ratio between $0.1$ and $0.3$ is frequently considered reasonable and can be achieved more easily (see, for example, \cite{sharpe1998sharpe}). The average Sharpe ratio is chosen to be 0.25 in our benchmark analysis. \\

\subsection{The impact of VaR on the terminal wealth}
\begin{figure}[t]	
	\centering
	\includegraphics[width=\textwidth]{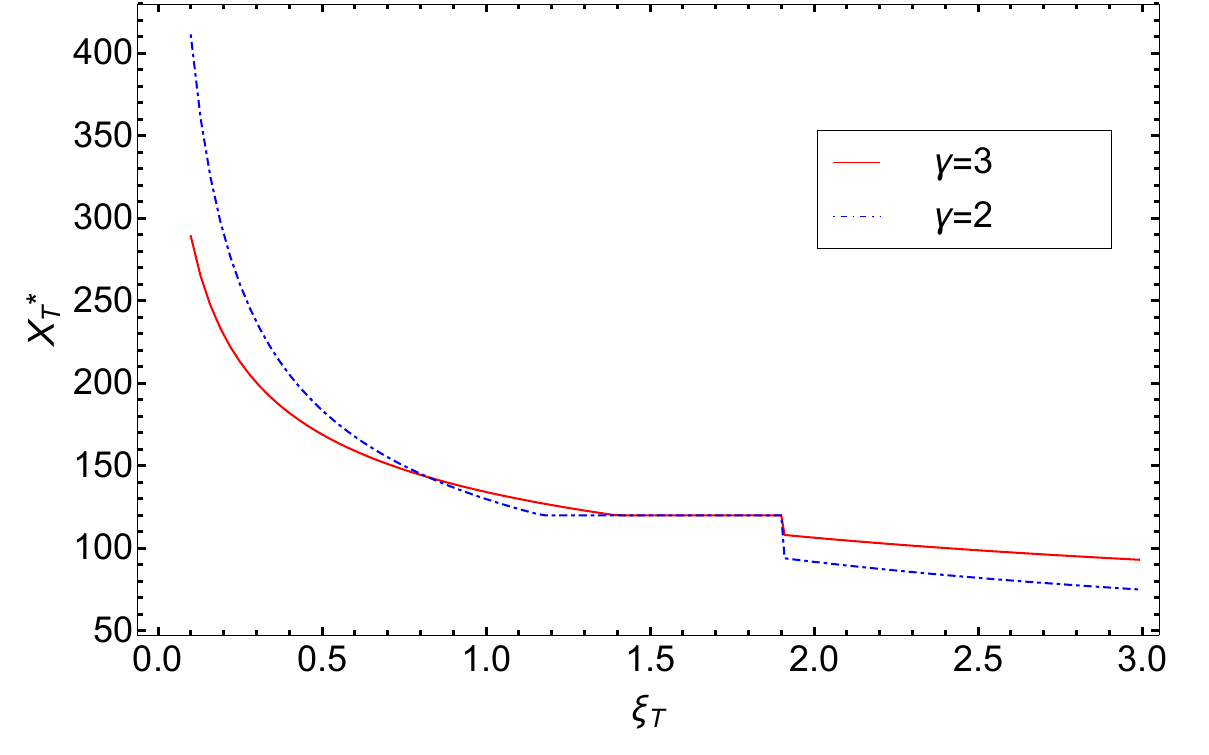}
	\caption{Optimal terminal wealth $X^*_T$ as a function of the state price density $\xi_T$ using parameters given in Section \ref{para_choice}.} \label{fig_opXT_partVaR}
\end{figure}

First of all, Figure \ref{fig_opXT_partVaR} demonstrates how the optimal final wealth $X^*_T$ changes with the state price density $\xi_T$, as typically done in the literature on utility maximization problems. A low state price density implies a rather rosy financial market, while a high state price density stands for an unfavorable performance of the market. 
The optimal final wealth $X^*_T$ is a decreasing function of $\xi_T$. As the market performance varies from a very favorable state (a very low $\xi_T$)
to a very unfavorable state (a very high $\xi_T$), the resulting
optimal final wealth moves from a very high to a rather low level.  With a VaR-type
constraint, the evolution of the optimal final wealth $X_T^*$
depends on the regulatory threshold $L$. The
resulting optimal wealth with a VaR-type constraint can be
decomposed into three parts, depending on two critical values of the state price densities, cf. Theorem \ref{theo:VaR}. 
In the good and bad performance areas, the
optimal wealth moves very similarly as typically in the Merton case. In the intermediate state, here in the interval $[1.42,1.91]$ for $\gamma=3$,  the investor hedges against losses in order to meet the regulatory threshold. This is due to the fact that
the VaR-type investor shall use these value reductions as
compensations to meet the regulatory threshold $L$ in the immediate region. In addition, we show the optimal solutions for two different levels of risk aversion $\gamma=2$ and $\gamma=3$. The higher critical state price density $\bar\xi$ does not depend on the risk aversion and is identical for both cases (it remains at the level of 1.91, while the lower critical value becomes 1.18 for $\gamma=2$.).  Comparing $\gamma=3$ to $\gamma=2$, a more risk-averse economic institution is interested in ensuring a higher payoff in really bad economic states. Consequently, in the good economic states, the more risk-averse institution obtains a lower wealth level, and the region in which the guaranteed payment is ensured becomes narrower.\\


\begin{figure}[t]
	\centering
	\includegraphics[width=\textwidth]{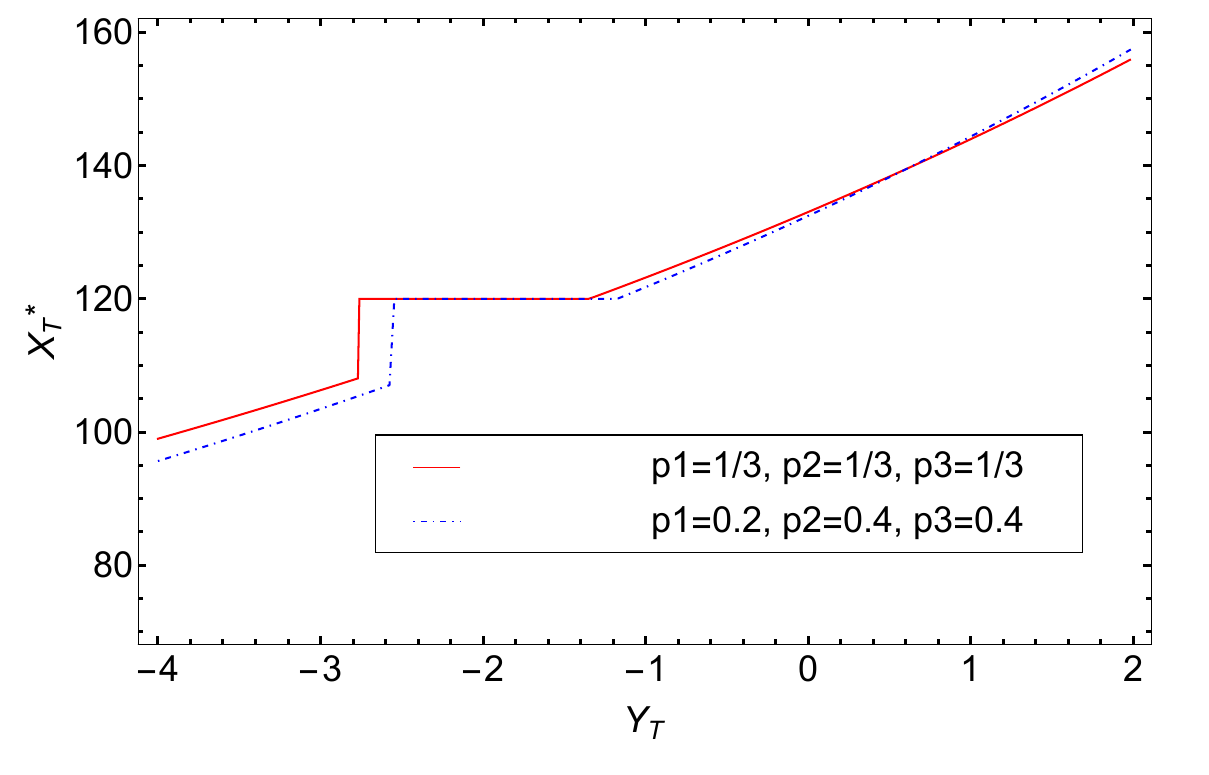}
	\caption{Optimal terminal wealth $X^*_T$ as a function of the state price density $Y_T$ using the parameters given in Section \ref{para_choice}. }
	\label{opXT2}
\end{figure}

\subsection{The impact of the prior for the VaR constraint}

To examine the influence of different prior probability distributions, we present the results in Figure \ref{opXT2}. This figure illustrates the optimal terminal wealth under a VaR constraint, as established by Theorem \ref{theo:VaR}. We consider two distinct sets of priors: $\{p_1=1/3, p_2=1/3, p_3=1/3\}$ and $\{p_1=0.2, p_2=0.4, p_3=0.4\}$. 
The market's risk prices, denoted as $\vartheta_1$, $\vartheta_2$, and $\vartheta_3$, remain consistent with those discussed in the previous subsection, where $\gamma=3$. To provide a precise analysis of the impact of these priors, we have chosen to plot the graph as a function of $Y_T$ instead of $\xi_T$, as the state price $\xi_T$ itself is influenced by the chosen priors. This approach allows us to isolate and thoroughly assess the effect of the prior probability distributions on the optimal terminal wealth. We can draw the following noteworthy observations: Initially, it's evident that as $Y_T$ rises, reflecting a more favorable economic state, the optimal terminal wealth $X_T^*$ also experiences an increase. This positive correlation aligns with economic intuition.
Secondly, it becomes apparent that the chosen priors play a significant role in determining the critical values $\bar\xi, \underline\xi$. These critical values determine specific regions within which the terminal wealth increases. When assigning greater weight to higher market prices of risks (the case with $\{p_1=0.2, p_2=0.4, p_3=0.4\}$), this  results in a shift of the guaranteed payment region to more favorable economic states.   


\section{Conclusion} \label{sec:con}

This paper analyzes the combined effects of partial information and (robust) VaR regulation on utility maximization problems. The inclusion of partial information does not change the shape of the optimal final wealth compared to the case of complete information, but it makes the solution more complex. The state price density process is much more complex in this case since the final state price density depends on the entire realization of the process up to the final time point. \\

\noindent If the regulator uses a different prior (or a worst-case prior) to set the regulatory constraint, the subjective belief of the regulator does not change the optimal structure, but affects the range in which the optimal terminal wealth matches the regulatory capital requirement.

\section{Appendix} \label{sec:app}

\subsection{Proof of Theorem \ref{theo:VaRoptstrategy}}\label{sec:proof1}
We know that 
\begin{align} 
X^*_{T}(\lambda_1,\xi_T) =  I(\lambda_1 \xi_T){\bf 1}_{\xi_T<\underline{\xi}} +L {\bf 1}_{\underline{\xi} \leqslant \xi_{T}< \bar{\xi}} + I(\lambda_1  \xi_T) {\bf 1}_{\bar{\xi}\le \xi_T},
\end{align} 
is the optimal terminal wealth for the specific parameters $\lambda_1, \underline{\xi}, \bar\xi.$ From this, we will first compute $X^*_t.$ We know that the discounted wealth process is a $\Qop$-martingale, i.e.
$$ \frac{X_t^*}{B(t)}=\Eop_\Qop\left[\frac{X_T^*(\lambda_1,\xi_T)}{B(T)} \Big| \mathcal{F}_t^S\right].$$

\noindent For general utility function, we can use a martingale representation theorem to obtain the existence of an $\mathcal{F}^S$-adapted process $\psi$ s.t.
\begin{equation}\label{eq:martrep}
    \frac{X_t^*}{B(t)}= x_0 + \int_0^t \psi_t dY_t 
\end{equation}
For the power utility we can compute $\Eop_\Qop\left[{X_T^*(\lambda_1,\xi_T)}| \mathcal{F}_t^S\right]$  explicitly as follows: First note that $y\mapsto F(t,y)$ is increasing and continuous and thus, admits an inverse function $F^{-1}.$
The expression $\Eop_\Qop\left[{X_T^*(\lambda_1,\xi_T)}| \mathcal{F}_t^S\right]$ can be split into three parts:

Part 1: In the expression below $Z$ is a standard normal random variable, independent from $Y_t.$
\begin{eqnarray*}
&& \Eop_\Qop\left[ I(\lambda_1 \xi_T){\bf 1}_{\xi_T<\underline{\xi}}    \big| \mathcal{F}_t^S\right] = 
\lambda_1^{-1/\gamma} e^{rT/\gamma} \Eop_\Qop\left[ F(T,Y(T))^{1/\gamma} {\bf 1}_{F(T,Y(T))>1/(e^{rT}\underline{\xi})}    \big| \mathcal{F}_t^S\right] \\
&=& \lambda_1^{-1/\gamma} e^{rT/\gamma} \Eop_\Qop\left[ F(T,Y(t)+\sqrt{T-t}Z)^{1/\gamma} {\bf 1}_{F(T,Y(t)+\sqrt{T-t}Z)>1/(e^{rT}\underline{\xi})}    \big| \mathcal{F}_t^S\right] \\
&=& \lambda_1^{-1/\gamma} e^{rT/\gamma} \Eop_\Qop\left[ F(T,Y(t)+\sqrt{T-t}Z)^{1/\gamma} {\bf 1}_{Z > \big(F^{-1}(T, 1/(e^{rT}\underline{\xi})-Y(t)\big)/  \sqrt{T-t} }\big| \mathcal{F}_t^S\right] \\
&=& \lambda_1^{-1/\gamma} e^{rT/\gamma} \int_{(F^{-1}(T,1/e^{rT}\underline{\xi}) -y)/\sqrt{T-t}}^\infty F(T,y+z\sqrt{T-t})^{1/\gamma}\varphi(z) dz = h_1(t,Y(t)).
\end{eqnarray*}
We proceed in the same way for the other two parts:

Part 2:
\begin{eqnarray*}
&& \Eop_\Qop\left[L {\bf 1}_{\underline{\xi} \leqslant \xi_{T}< \bar{\xi}}    \big| \mathcal{F}_t^S\right] = 
L \int_{(F^{-1}(T,1/e^{rT}\bar{\xi}) -y)/\sqrt{T-t}}^{(F^{-1}(T,1/e^{rT}\underline{\xi}) -y)/\sqrt{T-t}}
\varphi(z) dz= h_2(t,Y(t)).
\end{eqnarray*}

Part 3:
\begin{eqnarray*}
&& \Eop_\Qop\left[ I(\lambda_1  \xi_T) {\bf 1}_{\bar{\xi}\le \xi_T}  \big| \mathcal{F}_t^S\right] =\lambda_1^{-1/\gamma} e^{rT/\gamma} \Eop_\Qop\left[ F(T,Y(T))^{1/\gamma} {\bf 1}_{F(T,Y(T))\le 1/(e^{rT}\bar{\xi})}    \big| \mathcal{F}_t^S\right] \\
&=& 
 \lambda_1^{-1/\gamma} e^{rT/\gamma} \int_{-\infty}^{(F^{-1}(T,1/e^{rT}\bar{\xi}) -y)/\sqrt{T-t}} F(T,y+z\sqrt{T-t})^{1/\gamma}\varphi(z) dz = h_3(t,Y(t)).
\end{eqnarray*}
In total, we obtain for the discounted optimal wealth process
\begin{equation}\label{eq:optimal_wealth_process}
    \hat X_t^* := \frac{X_t^*}{B(t)} = e^{-rT} \Big( h^1(t,Y(t))+h^2(t,Y(t))+h^3(t,Y(t))\Big)=h(t,Y(t)).
\end{equation}
On the other hand the discounted wealth process $\hat X$ of an arbitrary strategy satisfies
$$ d\hat X(t) = \frac{\pi(t)}{B(t)}\sigma dY(t).$$
Equating this with \eqref{eq:martrep} gives for general utility $\pi^*(t) = \psi_t B(t)/\sigma.$ For the power utility we obtain by applying the It\^{o} formula to \eqref{eq:optimal_wealth_process}  since $\hat X^*$ is a $\Qop$-martingale
$$d\hat X^*_t = 0 dt + h_y(t,Y(t))dY(t) = h_y(t,Y(t))dY(t). $$ 
Equating the expressions and solving for $\pi^*(t)$ yields the statement for power utility.

\subsection{Proof of Theorem \ref{theo:VaR2}} \label{sec:proof2}
To simplify the notation, we write $X_T^*$ instead of $X_T^*(\lambda_1,\lambda_2,\xi_T,\eta_T).$
Since 
$$\tilde\Pop(X_T^* \ge L) = \tilde \Eop[1_{[X_T^*\ge L]}]= \Eop[\tilde \eta_T1_{[X_T^*\ge L]} ] = \Eop[ \Eop[\tilde\eta_T |\FC_T^S] 1_{[X_T^* \ge L]}]=\Eop[\eta_T 1_{[X_T^* \ge L]}] $$  the Lagrange function of problem \eqref{eq:optproblem2} is given by
$$ L(X,\lambda_1,\lambda_2) := U(X)-\lambda_1\xi_T X+\lambda_2 \eta_T 1_{[X\ge L]}.$$
 Note that $I(\lambda_1 \xi_T)$ is the maximum point of $x\mapsto U(x)-\lambda_1 \xi_T x$ and thus if it satisfies the probability constraint is the optimal solution. 
 
\noindent  For fixed $z_1,z_2>0$ we consider the problem to maximize (point-wise)
 $$ X \mapsto U(X)-z_1 X+z_2 1_{[X\ge L]}.$$
The maximum point of the Lagrange function can either be at $I(z_1)$ or at $L$. We show that it is given by
$$ X^* =  I(z_1){\bf 1}_{z_1<U'(L)} +L {\bf 1}_{U'(L) \leqslant z_1< \bar\xi} + I(z_1) {\bf 1}_{\bar\xi\le z_1}$$
where $\bar\xi := \inf\{ z>U'(L) : f(z)=z_2+U(L)\}.$ 
Before we start note that
the function $f(x) = U(I( x)) - x I( x)+ Lx$ is strictly increasing on $x>U'(L)$ which can be seen by inspecting its derivative and $f(U'(L))=U(L).$ Since $\inf\emptyset := +\infty,$  we may have $\bar \xi=+\infty$ in which case the last part of $X^*$ can be skipped. 
Now consider the following cases:
\begin{itemize}
    \item[1.]  $z_1\le U'(L)$.
    
    Hence $X^*=I(z_1) \ge L $ and thus $I(z_1)$  is the maximum point.
    \item[2.] $U'(L)< z_1 \le \bar \xi.$
    
     Hence $I(z_1) < L.$ Also note that $U'(L)< \bar \xi. $ In this case
     \begin{eqnarray*}
     U(L) -z_1 L + z_2 &>&   U(I(z_1)) -z_1 I(z_1) \Leftrightarrow f(z_1) < U(L)+z_2.
     \end{eqnarray*}
     Hence $L$ is the maximum point.
     \item[3.] $z_1 > \bar \xi.$
     
     In this case the inequality in the previous case is reversed and $I(z_1)$  is the maximum point.
\end{itemize}
Finally let $X_T^*$ be as stated in Theorem \ref{theo:VaR2} and let $X_T$ be any other admissible terminal wealth. Then
  \begin{eqnarray*}
   && \Eop U(X_T^*)- \Eop U(X_T) =  \Eop U(X_T^*)- \Eop U(X_T) -\lambda_1 x_0 +\lambda_1 x_0 +\lambda_2 (1-\beta) - \lambda_2 (1-\beta)\\
    &\ge &  \Eop U(X_T^*)- \Eop U(X_T) -\Eop[\lambda_1 \xi_T X_T^*]+ \Eop[\lambda_1 \xi_T X_T]+\Eop[\lambda_2\eta_T 1_{[X_T^* \ge L]}]-\Eop[\lambda_2\eta_T 1_{[X_T \ge L]}]\ge 0
     \end{eqnarray*}
due to our previous discussion. 
Thus, $X_T^*$ has the maximal expected utility. Since it is also admissible, the statement is shown.
     
Note that for $\eta_T\equiv 1$, we can define $$\lambda_2= U(I( \lambda_1\bar\xi)) -\lambda_1\bar \xi I( \lambda_1\bar\xi) + \lambda_1\bar \xi L -U(L) $$  and $\bar\xi$ such that $\Pop(\xi_T>\bar \xi)=\beta$ which implies that $\bar \xi$ is deterministic here.

\subsection{Proof of Corollary \ref{cor:1}}\label{sec:proof3}
We will show that the optimal solution of Theorem \ref{theo:VaR} also satisfies the constraints under the different priors.

Since $y\mapsto F(t,y)$ is increasing in this case (note that all $ \vartheta_i$ are non-negative) we obtain due to our ordering of the values that 
$$\Pop\Big(\frac{e^{-rT}}{F(T,W_T+\vartheta_1T)} \ge \bar \xi\Big)\ge \ldots\ge \Pop\Big(\frac{e^{-rT}}{F(T,W_T+\vartheta_mT)} \ge \bar \xi\Big). $$
Hence $$f(k) := \Pop\Big(\frac{e^{-rT}}{F(T,W_T+\vartheta_kT)} \ge \bar \xi\Big)$$
is decreasing in $k$ and thus when $\theta\le_{FSD}\theta'$ we have $$\tilde\Pop(\xi_T \ge \bar \xi) = \sum_{k=1}^m q_k f(k) \ge \sum_{k=1}^m q'_k f(k) =\tilde\Pop'(\xi_T \ge \bar \xi).$$
Thus, if $\bar \xi$ satisfies $\Pop(\xi_T \ge \bar \xi)\le \beta$ in Theorem \ref{theo:VaR},  the value $\bar\xi$ also satisfies the constraints $\tilde\Pop_i(\xi_T \ge \bar \xi)\le \beta$ and is thus optimal under these different beliefs. 

\singlespacing
\bibliographystyle{apalike}
\bibliography{bibfile}

\end{document}